\RequirePackage[l2tabu, orthodox]{nag}		
\documentclass[10pt]{article}

\usepackage[T1]{fontenc}
\usepackage{lmodern}
\usepackage{natbib}
\usepackage{graphicx}
\usepackage{caption}
\usepackage{subcaption}
\usepackage{microtype}
\usepackage[frenchmath]{mathastext}					
\usepackage{amsmath}                        
\numberwithin{equation}{section}
\usepackage{amsfonts}
\usepackage{graphicx}                      

\usepackage{natbib}    
\setcitestyle{authoryear}
\usepackage[plainpages=false, pdfpagelabels]{hyperref} 
	\hypersetup{
		colorlinks   = true,
		citecolor    = RoyalBlue, 
		linkcolor    = RubineRed, 
		urlcolor     = Turquoise
	}
\usepackage[dvipsnames]{xcolor}
\usepackage[all]{xy}
\usepackage{amssymb}                        
\usepackage[mathscr]{eucal}                 
\usepackage{dsfont}													
\usepackage[paperwidth=8.5in,paperheight=11.00in,top=1.25in, bottom=1.25in, left=1.00in, right=1.00in]{geometry}
\usepackage{mathtools}                      
\mathtoolsset{showonlyrefs=true}            
\usepackage{amsmath}               
\MakeRobust{\eqref}
\usepackage{microtype}											
\usepackage{amsthm}                         
\allowdisplaybreaks                         
\theoremstyle{plain}
\newtheorem{theorem}{Theorem}
\numberwithin{theorem}{section}
\newtheorem{lemma}[theorem]{Lemma}          
\newtheorem{proposition}[theorem]{Proposition}

\theoremstyle{definition}

\renewcommand{\[}{\left[}

\newcommand\Eb{\mathds{E}}
\newcommand\Fb{\mathds{F}}

\newcommand\Pb{\mathds{P}}

\newcommand\Rb{\mathds{R}}

\newcommand\Ac{\mathscr{A}}

\newcommand\Cc{\mathscr{C}}

\newcommand\Fc{\mathscr{F}}

\newcommand\Ic{\mathscr{I}}

\newcommand\Oc{\mathscr{O}}

\newcommand\Tc{\mathscr{T}}

\newcommand\om{\omega}
\newcommand\Om{\Omega}
\newcommand\sig{\sigma}

\newcommand\Lam{\Lambda}
\newcommand\gam{\gamma}
\newcommand\Gam{\Gamma}

\newcommand\del{\delta}

\newcommand\hh{\widehat{h}}

\newcommand\Jt{\widetilde{J}}

\newcommand\ut{\widetilde{u}}

\newcommand\taut{\widetilde{\tau}}

\newcommand\Lamt{\widetilde{\Lam}}

\renewcommand\d{\partial}

\newcommand\dd{\mathrm{d}}
\newcommand\ee{\mathrm{e}}

\begin{document}

\title{Optimal times to buy and sell a home}

\author{
Matthew Lorig
\thanks{
Department of Applied Mathematics, University of Washington.
\textbf{e-mail}: \url{mlorig@uw.edu}}
\and
Natchanon Suaysom
\thanks{
Department of Applied Mathematics, University of Washington.
\textbf{e-mail}: \url{nsuaysom@uw.edu}}
}

\date{This version: \today}

\noindent

\maketitle



\begin{abstract}
We consider a financial market in which the risk-free rate of interest is modeled as a Markov diffusion.
We suppose that home prices are set by a representative home-buyer, who can afford to pay only a fixed cash-flow per unit time for housing.
The cash-flow is a fraction of the representative home-buyer's salary, which grows at a rate that is proportional to the risk-free rate of interest.
As a result, in the long-run, higher interest rates lead to faster growth of home prices.
The representative home-buyer finances the purchase of a home by taking out a mortgage.
The mortgage rate paid by the home-buyer is fixed at the time of purchase and equal to the risk-free rate of interest plus a positive constant.
As the home-buyer can only afford to pay a fixed cash-flow per unit time, a higher mortgage rate limits the size of the loan the home-buyer {can} take out.
As a result, the short-term effect of higher interest rates is to lower the value of homes.
In this setting, we consider an investor who wishes to buy and then sell a home in order to maximize his discounted expected profit.  
This leads to a nested optimal stopping problem.  
We use a nonnegative concave majorant approach to derive the investor's optimal buying and selling strategies.
Additionally, we provide a detailed analytic and numerical study of the case in which the risk-free rate of interest is modeled by a Cox-Ingersoll-Ross (CIR) process.
We also examine, in the case of CIR interest rates, the expected time {that the} investor waits before buying and then selling a home when following the optimal strategies.
\end{abstract}

\noindent
{\,}
\\[1em]

%
%

\section{Introduction}
\label{sec:introduction}
While many consider a home merely as a place to live, it is also financial asset, the purchase and subsequent sale of which can generate a significant profit.  
The problem of buying and/or selling a home in order to minimize purchase price and/or maximize sale price or profit has been widely studied in academic literature. Various mathematical tools have been used to solve this problem, including multivariate probability theory, game theory, and optimal stopping theory.  For example, \cite{bruss1997multiple} assumes home prices follow a specific probability distribution, and derives the optimal stopping rules for buying and selling homes. \cite{anglin2004long} derives an optimal stopping strategy from the perspective of a representative-home buyer who is observing multiple other homes to purchase. \cite{egozcue2013optimal} and \cite{brown2013role} use an optimal stopping approach to optimize the profit in a home bidding process. \cite{albrecht2016directed} considers house buying and selling in game theoretic framework and derive prices at equilibrium. And \cite{leung2017flipping} use a housing market search model to derive prices at equilibrium.
\\[0.5em]
When deriving the optimal home-buying or home-selling strategy, one must consider a number of factors such as, e.g., interest rates, transaction costs, an investor's discount rate, the demand and supply of homes in certain locations, quality of nearby schools, etc..  Among the many factors one might consider, perhaps the most important is the interest rate. To illustrate the important role interest rates play in home prices, in Figure \ref{fig:house-rate-price}, we plot the Monthly S\&P/Case-Shiller U.S. National Home Price Index 
and the Weekly 30-Year Fixed Rate Mortgage Average in the United States 
from January 2019 to January 2021.  The figure clearly shows during this period that home prices were inversely related to interest rates.  This data is consistent with the theoretical results of \cite{leung2017flipping} who show that, in equilibrium, home {prices} are inversely related to interest rates. 
%
\\[0.5em]
In the present paper, we present a framework {that describes} how the risk-free rate of interest affects home prices.
Briefly, we suppose that home prices are set by a representative home-buyer, who can afford to pay only a fixed cash-flow per unit time for housing.
The cash-flow is a fraction of the representative home-buyer's salary, which grows at a rate that is proportional to the risk-free rate of interest.
As a result, in the long-run, higher interest rates lead to faster growth of home prices.
The representative home-buyer finances the purchase of a home by taking out a mortgage.
The mortgage rate paid by the home-buyer is fixed at the time of purchase and equal to the risk-free rate of interest plus a positive constant.
As the home-buyer can only afford to pay a fixed cash-flow per unit time, a higher mortgage rate limits the size of the loan the home-buyer {can} take out.
As a result, the short-term effect of higher interest rates is to lower the value of homes.
In this setting, we consider an investor {that} wishes to maximize his expected discounted profit from buying a home and selling it at a later time.  
As the optimal time to buy a home depends on the optimal time to sell a home, this leads to a \textit{nested optimal stopping problem}.
The main purpose of this paper is to solve this nested optimal stopping problem by providing an explicit characterization of the optimal buying and selling times when the risk-free rate of interest is modeled as a Markov diffusion and to provide a detailed study of the case in which the risk-free rate of interest is modeled as a Cox-Ingersoll-Ross (CIR) process.
\\[0.5em]
Mathematically, our problem formulation 
falls within a class of optimal stopping problems with {stochastic} discounting studied in \cite{dayanik2008optimal}. To obtain the investor's value function, we use the nonnegative concave majorant approach developed by \cite{dayanik2008optimal}.  This approach has been applied to a variety of optimal stopping problems. For instance, \cite{leung2014optimal} uses this approach to derive optimal strategies for the problem of starting-stopping a CIR process. And \cite{leung2015optimal} uses this approach to derive the optimal timing for trading with transaction costs where the trading price spread between two assets is modeled by an Ornstein–Uhlenbeck (OU) process.
\\[0.5em]
The rest of this paper proceeds as follows.
In Section \ref{sec:interest} we present a model for how the risk-free rate of interest affects home values.
Next, in Section \ref{sec:optimal}, we define the investor's optimal home-buying and home-selling problems.
The optimal home-buying and home-selling problems fall into a larger class of optimal stopping problems with stochastic discounting.
We provide a general solution to these optimal stopping problems in Section \ref{sec:J}.
In Section \ref{sec:cir} we focus specifically on the case in which the risk-free rate of interest is described by a CIR process.
We derive expressions for the value functions and optimal stopping times that correspond to the investor's optimal buying and selling problems.
Additionally, we calculate the expected time the investor waits to buy and then holds a home before selling, assuming he follows the optimal buying and selling strategies.  
Lastly, in Section \ref{sec:conclusion}, we offer some thoughts on future directions of research.


\section{The relation between interest rates and home values}
\label{sec:interest}
Throughout this paper, we fix a probability space $(\Om, \Fc, \Pb)$ and a filtration $\Fb = (\Fc_t)_{t \geq 0}$. The probability measure $\Pb$ represents the real world probability measure. In this setting, let $R = (R_t)_{t\geq 0}$ denotes the risk-free rate of interest.   We shall suppose that $R$ is a regular diffusion that lives on an interval $\Ic := (x,y)$, where the end points $x$ and $y$ are natural and satisfy $0 \leq x < y \leq \infty$.  Specifically, we suppose that $R$ is the unique strong solution to a stochastic differential equation (SDE) that is of the form
\begin{align}
    \dd R_t & = \mu(R_t) \dd t + \sig(R_t) \dd W_t, \label{eq:R-def}
\end{align} 
where $W = (W_t)_{t \geq 0}$ is a one-dimensional $(\Pb,\Fb)$-Brownian motion and the functions $\mu$ and $\sig$ satisfy
\begin{align}
\mu
	&: \Ic \to \Rb , &
\sig
	&: \Ic \to \Rb_{++} ,
\end{align}
with $\Rb_{++} := (0,\infty)$.
\\[0.5em]
The aim of this section is to develop a framework that captures how the dynamics of $R$ affect home values.  To this end, we consider a representative home-buyer who at time $t$ can afford to pay a cash flow of $(C_t)_{t\geq 0}$ per unit time for housing.  As time passes, the home-buyer's wages will increase and, as such, so will the amount of money he can afford to pay for housing.  To capture this effect, we suppose that the dynamics of the cash flows are as follows
\begin{align}
C_t
	&= C \ee^{ \gam \int_0^t R_s \dd s } , &
C
	&> 0 , &
\gam
	&> 0 . \label{eq:C}
\end{align}
Equation \eqref{eq:C} assumes that the amount of money the representative home-buyer can allocate to housing per unit time grows at a rate $\gam R$ that is proportional to the risk-free rate of interest.  If one considers the risk-free rate $R$ to be a proxy for inflation, then $\gamma$ captures how quickly the home-buyer's wages grow in real (as opposed to nominal) terms.  If $\gam > 1$ the home-buyer's wages grow faster than inflation and he is getting richer over time.  On the other hand, if $\gam < 1$ the home-buyer's wages are not keeping up with inflation and, over time, he is becoming poorer.
\\[0.5em]
Now, suppose that, at time $t$, the representative home-buyer has found a home he wishes to purchase.  In order to finance this purchase, he takes out a loan from a bank with a repayment period of $T$ years at a fixed interest rate $R_t + \rho$ where $\rho > 0$.  The constant $\rho$ captures the fact that home-buyer may default on his loan payments and, thus, should be charged an interest rate that is higher than the risk-free rate of interest.  As, at time $t$, the representative home-buyer can only afford to pay a cash-flow of $C_t$ per unit time, the maximum value of the home he can afford is
\begin{align}
\int_t^{t+T} C_t \ee^{ - (R_t + \rho) (u - t) } \dd u
	&=	\frac{C_t}{R_t + \rho} \Big( 1 - \ee^{- (R_t + \rho)T} \Big) .
\end{align}
Although home-buyers of different economic classes will be able to afford different cash-flows for housing, the relationship between the value of a home and the interest rate $R$ will be the same for all homes in the economy.  Thus, the value  $V = (V_t)_{t \geq 0}$ of any homes in the economy is given by 
\begin{align}
V_t
	&=			v(R_t) \ee^{ \gam \int_0^t R_s \dd s } , &
v(R_t)
	&:= \frac{C}{R_t + \rho} \Big( 1 - \ee^{- (R_t + \rho)T} \Big) , \label{eq:V}
\end{align}
where $C$ is a constant that captures the relative expense of a particular home; it will play no role in the analysis that follows.  It is important to notice that the interest rate $R$ has both a long-term and a short-term effect on the value $V$ of a home.  In the long-term, higher {interest} rates have the effect of raising the value of a home due to the term $\ee^{ \gam \int_0^t R_s \dd s }$.  In the short-term, the effect of interest rates on home values is captured by $v(R_t)$.  Using the fact that $\ee^x > 1+x$ for any $x > 0$, we have that
\begin{align}
    v'(r) & = - \frac{C \ee^{- (r+\rho)T}}{(r+\rho)^2} \left(\ee^{ (r+\rho)T} - 1 - (r+\rho)T \right) < 0. \label{eq:v-diff}
\end{align}
This means that $v(r)$ is a decreasing function of $r$, and that in the short-term, higher interest rates have the effect of lowering the value of a home. The dynamics of $V$ is given by
\begin{align}
    \dd V_t & = \bigg( \gam R_t +\frac{1}{v(R_t)}\Big(\mu(R_t) v'(R_t) + \frac{1}{2}\sig^2(R_t) v''(R_t)\Big) \bigg)V_t \dd t +  \frac{v'(R_t)\sig(R_t)}{v(R_t)} V_t \dd W_t.
\end{align}
Note that while $V$ alone is not a Markov process, the pair $(R,V)$ is Markov.

\section{Optimal home buying and selling problems}
\label{sec:optimal}
Having described the relationship between the risk-free rate of interest $R$ and home values $V$, we now consider an investor who wishes to buy and then sell a home in order to maximize the present value of these transactions.  
Note that, as short-selling of homes is not allowed, we will not consider cases in which the investor first sells and then later buys back a home.
We will suppose that for a payment $P_t$ received at time $t$ the investor assigns a present value of $\Eb ( \ee^{-\chi \int_0^t R_s \dd s} P_t )$, where $\chi > 0$ is a discount rate that is specific to the investor.  The larger the value of $\chi$, the more heavily the investor discounts future payments. 
One can alternatively consider constant discounting of the form $\Eb ( \ee^{-\chi t} P_t )$.  This case is discussed in Appendix \ref{sec:constant-discount}.
\\[0.5em]
Let us denote by $\tau_b$ and $\tau_s$, respectively, the times at which the investor buys and sells a home. 
In general, $\tau_b$ and $\tau_s$ will be (random) $\Fb$-stopping times.  
Because the investor is not purchasing a primary residence, the interest rate he would pay were he to take out a loan for a home would be very high.  As such, we will suppose that the investor pays cash for a home.  The amount of money the investor will need to pay at time $\tau_b$ to buy a home will be
\begin{align}
\text{Cost of home purchase}
	&=	V_{\tau_b} (1+\del_b) + K_{b,\tau_b}, &
\del_b 
	&> 0 , &
K_{b,t} 
	&:=  K_b \ee^{ \gam \int_0^{t} R_s \dd s } , &
K_b
	&> 0 , \label{eq:costs}
\end{align}
where $\del_b$ represents a transaction cost that is proportional to the value of a home price (e.g., a fee to a realtor) and $K_{b,\tau_b}$ represents fixed transaction costs (e.g., fees paid to a title company).  Note that the fixed transaction cost $K_{b,\tau_b}$ grows over time due to inflation whereas the proportional transaction cost $\del_b V_{\tau_b}$ scales with the value of a home.  Similarly, when the investor sells a home he has purchased, he will receive
\begin{align}
\text{Revenue from home sale}
	&=	V_{\tau_s} (1-\del_s) - K_{s,\tau_s}, &
\del_s 
	&> 0 , &
K_{s,t} 
	&:=  K_s \ee^{ \gam \int_0^{t} R_s \dd s } , &
K_s
	&> 0 , \label{eq:revenue}
\end{align}
where $\del_s$ and $K_{s,\tau_s}$ capture proportional and fixed transaction costs, respectively.
\\[0.5em]
Although chronologically, the investor must buy a home before he sells it, we will consider the optimal selling problem first.
Let $\Tc$ be the set of $\Fb$-stopping times.
For a fixed selling strategy $\tau_s \in \Tc$ the expected discounted revenue the investor will receive from selling the home is
\begin{align}
J_s^{\tau_s}(r)
	&:= \Eb \Big[
			\ee^{-\chi \int_0^{\tau_s} R_s \dd s} (V_{\tau_s}  (1 - \del_s) - K_{s,\tau_s}) \Big| R_0 = r \Big] .
\end{align}
Recalling the relationship \eqref{eq:V} between $V$ and $R$, and introducing the process $\Lam = (\Lam_t)_{t \geq 0}$, defined by
\begin{align}
\Lam_t 
	&:= (\chi-\gam)\int_{0}^{t} \, R_s \dd s , \label{eq:Lambda}
\end{align}
we can re-write $J_s^{\tau_s}(r)$ more compactly as follows
\begin{align}
J_s^{\tau_s}(r)
	&=	\Eb \Big[ \ee^{-\Lam_{\tau_s}} f_s(R_{\tau_s}) \Big| R_0 = r \Big] , &
f_s(r)
	&:=  v(r)(1-\del_s)-K_s . \label{eq:fs-def}
\end{align}
In order to maximize the present value of the revenue received from selling a home, the investor will need to maximize {$J_s^{\tau_s}$} over all stopping times $\tau_s \in \Tc$.
We therefore define the \textit{selling value function} $J_s$ and \textit{optimal selling strategy} $\tau_s^*$ (assuming it exists) as follows
\begin{align}
J_s(r)
	&:=	\sup_{\tau_s \in \Tc} {J_s^{\tau_s}(r)} =: J_s^{\tau_s^*}(r) . \label{eq:Js-def}
\end{align}
Now, let us assume that the investor will follow the optimal selling strategy $\tau_s^*$.  Then, for a fixed buying strategy $\tau_b$, the expected discounted profit he will receive from buying and then selling a home is given by
\begin{align}
J_b^{\tau_b}(r)
	&:=	\Eb \Big[
			\ee^{-\chi \int_0^{\tau_b} R_s \dd s} J_s(R_{\tau_b}) 
			- \ee^{-\chi \int_0^{\tau_b} R_s \dd s} (V_{\tau_b}  (1+\del_b)+K_{b,\tau_b}) 
			\Big| R_0 = r \Big] .
\end{align}
Recalling the relationship \eqref{eq:V} between $V$ and $R$, the definition  \eqref{eq:Lambda} of $\Lambda$ and the definition \eqref{eq:Js-def} of $J_s$, we can express $J_b^{\tau_b}(r)$ more compactly as follows
\begin{align}
{J_b^{\tau_b}(r)}
	&=	\Eb \Big[  \ee^{-\Lam_{\tau_b}} f_b( R_{\tau_b} ) \Big| R_0 = r \Big ] , &
f_b(r) 
	& :=  J_s(r) - \big(v(r)(1+\del_b) + K_b \big) . \label{eq:fb-def}
\end{align}
In order to maximize the present value of the purchase and sale of a home, the investor will need to maximize $J_b^{\tau_b}$ over all stopping times $\tau_b \in \Tc$.
We therefore define the \textit{buying value function} $J_b$ and the \textit{optimal buying strategy} $\tau_b^*$ (assuming it exists) as follows
\begin{align}
J_b(r)
	&:=	\sup_{\tau_b \in \Tc} {J_b^{\tau_b}(r) =: J_b^{\tau_b^*}(r) .} \label{eq:Jb-def}
\end{align}
Note that $J_s$ and $J_b$ are the special cases of a class of \textit{optimal stopping problems with {stochastic} discounting} of the form
\begin{align}
    J(r) &:= \sup_{\tau \in \Tc} J^\tau(r) =: J^{\tau^*}(r) , &
		J^\tau(r) &:= \Eb \Big[\ee^{-\Lam_{\tau}} f(R_{\tau}) \Big| R_0 = r \Big] . \label{eq:J-def}
\end{align}
Note also that, in order for a nontrivial optimal stopping time of \eqref{eq:J-def} to exist, we must have $\Lam > 0$.  Thus, we assume that $\chi > \gam$ throughout this paper.  We shall refer to $J$ and $\tau^*$ (with no subscripts) as the \textit{value function} and \textit{optimal stopping time}, respectively. For ease of notation, in the sections that follow, we will use $J$ in an expression that holds true for both $J_b$ and $J_s$, $\tau^*$ in an expression that holds true for $\tau_b^*$ and $\tau_s^*$ and $f$ in an expression that holds true for $f_b$ and $f_s$.
\\[0.5em]
Before deriving explicit characterizations of the optimal buying and selling times, let us examine qualitatively what $\tau^*_b$ and $\tau^*_s$ should look like.  Recall that the short-term effect of the risk free rate of interest $R$ on home prices is captured by $v(R_t)$, which is a decreasing function of $R_t$.
As the investor will want to buy a home when prices are relatively low, we expect that the optimal buying strategy $\tau_b^*$ will involve waiting until interest rates $R$ rise to a value $r_b$ called the \textit{buying threshold}.  Similarly, as the investor will want to sell when home prices are relatively high, we expect that the optimal selling strategy $\tau_s^*$ will involve waiting until the risk-free rate of interest $R$ falls to a value $r_s$ call the \textit{selling threshold}, where $r_s < r_b$.  In other words, we expect the optimal buying and selling strategies to be of the form
\begin{align}
\tau_b^*
    &:= \inf\{ t \geq 0 : R_t \geq r_b \} , &
\tau_s^*
    &:= \inf\{ t \geq 0 : R_t \leq r_s \} , \label{eq:tau-intuitive}
\end{align}
where $x < r_s < r_b < y$.

\section{Expressions for the value function $J$ and optimal stopping time $\tau^*$}
\label{sec:J}
In this section, we present the expressions for the value function $J$ and optimal stooping time $\tau^*$, which are defined in \eqref{eq:J-def}.  The expressions can be applied to the optimal selling problem \eqref{eq:Js-def} and optimal buying problem \eqref{eq:Jb-def}.
\\[0.5em]
To begin, let $\Ac$ {denote} the infinitesimal generator of the risk-free rate of interest process $R$. We have 
\begin{align}
\Ac
	&=	\mu(r) \d_r + \tfrac{1}{2} \sig^2(r) \d_r^2 . \label{eq:A-def}
\end{align}
Consider the following ordinary differential equation (ODE) for a function $u: \Ic \to \Rb$
 \begin{align}
    \Big(\Ac - (\chi-\gam)r \Big) u(r) 
			& = 0 . \label{eq:uode}
\end{align}
Suppose that \eqref{eq:uode} has two independent solutions $u = (u_+,u_-)$ such that $u_+$ is positive and strictly increasing and $u_-$ is positive and strictly decreasing.  It is well-known (see \cite[Equation (5)]{dayanik2008optimal}, for instance) that the functions $u_+$ and $u_-$ are related to the hitting times of the process $R$ as follows
\begin{align}
    \Eb\Big[\ee^{-\Lam_{\tau_c}} {\Big|} R_0 = r\Big] & = \begin{cases}
        u_+(r)/u_+(c) & r \leq c
        \\
        u_-(r)/u_-(c) & r > c
    \end{cases} , &
\tau_c
	&:= \inf\{ t \geq 0 : R_t = c \} ,
    \label{eq:psiphi_stopping}
\end{align}
where $c,r \in \Ic$.
\\[0.5em]
Next, we define the functions $g : \Ic \to g(\Ic)$ and $h : g(\Ic) \to \Rb$, which will be used in the expression of $J$ by
\begin{align}
    g(r) &:= -\frac{u_-(r)}{u_+(r)},\quad r\in \Ic, & 
    h(q) &:=  \frac{f(g^{-1}(q))}{u_+(g^{-1}(q))}, \quad  q \in g(\Ic) . \label{eq:hdef}
\end{align}
We define $h_b$ and $h_s$ from $f_b$ and $f_s$, respectively, in the same way we define $h$ from $f$.  To ease the notation, we use $h$ to represent expressions that hold true for both $h_b$ and $h_s$. Because $u_+$ is strictly positive increasing and $u_-$ is strictly positive decreasing, $g$ is strictly negative increasing, which means that $g^{-1}$ is well defined. The following proposition shows that the value function $J$ can be written in terms of $u_+$, $u_-$, $g$, and $h$.
\begin{proposition}
\label{th:J}
Suppose that the risk-free rate of interest $R$ is defined by \eqref{eq:R-def} on an interval $\Ic = (x,y)$ where $x$ and $y$ are natural boundaries. Let the functions $f$, {$u_+$, $u_-$}, $g$ and $h$ be as defined in \eqref{eq:J-def}, \eqref{eq:uode} and \eqref{eq:hdef}. If  both of the following limits are finite
\begin{align}
    \ell_x & := \lim_{r \to x^+} \frac{f^+(r)}{u_-(r)} , & 
		\ell_y & := \lim_{r \to y^-} \frac{f^+(r)}{u_+(r)} , &
		f^+(r) &:= \max(f(r),0) ,
\end{align}
then the value function $J$ defined in \eqref{eq:J-def} can be written as
\begin{align}
J(r) & =  u_+(r)\hh\big(g(r)\big), & r \in \Ic, \label{eq:Jsol}
\end{align}
where $\hh$ is the \emph{smallest decreasing nonnegative concave majorant (NCM)} of $h$.
\end{proposition}
\begin{proof}
See \cite[Proposition 3.4]{dayanik2008optimal}.
\end{proof}

\noindent
It is well-known (see, for instance \cite[Appendix D]{shreve1994optimal}) that the optimal stopping time $\tau^*$ can be computed from $J$ as follows
\begin{align}
\tau^* &  := \inf \{t \geq 0 : R_t \notin \Cc  \} , &
	&\text{where}&
\Cc & := \{r \in \Ic : J(r) > f(r)\} . \label{eq:tau_sol}
\end{align}
We refer to the set $\Cc$ as the \textit{continuation region}.

\section{Detailed analysis: CIR process risk-free rate}
\label{sec:cir}
In this section, we derive the expressions of $J$ and $\tau^*$ when the risk-free rate of interest is modeled by a CIR process. Specifically, suppose that the dynamics of risk-free rate of interest $R$ is given by 
\begin{align}
\dd R_t
	&=	\kappa (\theta - R_t ) \dd t + \sig \sqrt{R_t} \dd W_t, \label{eq:CIR}
\end{align}
where $\kappa,\theta,\sig > 0$.  We shall assume the Feller's condition $2 \kappa \theta > \sig^2$ is satisfied, which guarantees that $R$ never reaches zero. Note that $R$ is regular on $\Ic = (0,\infty)$ and both boundaries $0$ and $\infty$ are natural. Using \eqref{eq:A-def}, the infinitesimal generator $\Ac$ of the CIR process is given by
\begin{align}
\Ac
	&=	\kappa (\theta - r ) \d_r + \tfrac{1}{2} \sig^2 r \d_r^2 . \label{eq:A-cir}
\end{align}
 Using this specific infinitesimal generator \eqref{eq:A-cir}, the ODE \eqref{eq:uode} can be written as
\begin{align}
    \Big(\kappa (\theta - r ) \d_r + \tfrac{1}{2} \sig^2 r \d_r^2-(\chi-\gam)r\Big) u(r) & = 0. \label{eq:uspecific}
\end{align}
In Appendix \ref{sec:solveu}, we derive explicit expressions for positive increasing and positive decreasing solutions, $u_+$ and $u_-$, of \eqref{eq:uspecific}, which are given by
\begin{align}
        u_+(r) & = \ee^{-\nu r}  M(\alpha,\beta,\zeta r),
        &
        u_-(r) & = \ee^{-\nu r}  U(\alpha,\beta,\zeta r),\label{eq:u12}
\end{align}
where {$(\alpha,\beta,\xi,\zeta,\nu)$} are defined as follows
\begin{align}
    \alpha & :=\frac{\kappa \theta  }{{\sigma ^2}} \left(1-\frac{\kappa }{\xi}\right),
    &
    \beta & := \frac{2\kappa \theta   }{\sigma ^2}, & \xi & := \sqrt{\kappa^2 + 2 \sig^2(\chi-\gam)},
    &
    \zeta & := \frac{2 \xi}{\sigma ^2}, & \nu & := \frac{\alpha\zeta}{\beta} = \frac{\xi-\kappa}{\sig^2}, \label{eq:u12_param}
\end{align}
and where  $M$ and $U$ are the confluent hypergeometric function of the first kind and  second kind, respectively, as defined in \eqref{eq:MU_def}.  As $\kappa, \theta, \sig > 0$ and $\chi > \gam$, all parameters in \eqref{eq:u12_param} are positive, which allows us to write the following limit properties of $u_+$ and $u_-$
\begin{align}
    \lim_{r\to 0^+} u_+(r) &= 0, & 
    \lim_{r \to \infty} u_+(r) &= \infty, & 
    \lim_{r\to 0^+} u_-(r)  &= \infty, &  
    \lim_{r \to \infty} u_-(r) &= 0 . \label{eq:ulim}
\end{align}
We will use the limits in \eqref{eq:ulim} to verify the limit conditions of Proposition \ref{th:J}.
\\[0.5em]
In order to apply Proposition \ref{th:J} to determine the expressions for the value function $J$, it is necessary to determine $\hh$, the NCM of $h$. To that end, we need to know the sign of the slope and convexity of $h$ throughout $g(\Ic)$. Using the definition of $h$ in \eqref{eq:hdef} directly, the first and second derivative of $h$ are given by (with the shorthand $r:=g^{-1}(q)$)
\begin{align}
        h'(q) & =  \frac{1}{g'(r)} \frac{u_+(r)f'(r)-u_+'(r)f(r)}{\big(u_+(r)\big)^2}, \label{eq:h-diff1}
        \\
        h''(q) & = \frac{2}{\sig^2 ru_+(r) \big(g'(r)\big)^2} \Big(\Ac - (\chi-\gam)r\Big)f(r). \label{eq:h-diff2}
\end{align}
Equations \eqref{eq:h-diff1} and \eqref{eq:h-diff2} will be used to identify the critical and inflection points of $h$, which will then be used to calculate $\hh$, the NCM of $h$. We now have the necessary tools to derive the expressions for buying and selling value functions.

\subsection{Optimal home selling problem}
Although chronologically the investor will have to buy a home before being able to sell it, the optimal selling problem must be solved before the optimal buying problem due to the fact that the form of $f_b$ in \eqref{eq:fb-def} requires having known the selling value function $J_s$.  To derive the expression of $J_s$, we first define $h_s$ from $f_s$ the same way we define $h$ from $f$ in \eqref{eq:hdef} by
\begin{align}
    h_s(q) & :=  \frac{f_s\big(g^{-1}(q)\big)}{u_+\big(g^{-1}(q)\big)}, &  q < 0. \label{eq:hs-def}
\end{align}
It is straightforward to check using \eqref{eq:fs-def} and \eqref{eq:ulim} that
\begin{align}  \lim_{r\to 0^+} \frac{f^+_s(r)}{u_-(r)} & = 0, &
\lim_{r\to \infty} \frac{f^+_s(r)}{u_+(r)} & = 0.
\end{align}
This shows that the limit conditions in Proposition \ref{th:J} are satisfied. Next, we need to identify $\hh_s$, the NCM of $h_s$, which is done in Appendix \ref{sec:hhs_cal}. We have from \eqref{eq:hhs-explicit} that
\begin{align}
\hh_s(q) & = \begin{cases}
        h_s(q),  & q \leq q_s
        \\
        q\frac{h_s(q_s)}{q_s}, & q > q_s
        \end{cases}, & q_s := g(r_s),
\end{align}
where the selling threshold $r_s$ is the unique positive solution to equation \eqref{eq:rs}, which we repeat here for the reader's convenience
\begin{align}
\frac{u_-'(r_s)}{u_-(r_s)} & = \frac{f'_s(r_s)}{f_s(r_s)}.
\end{align}
Having confirmed that the limit conditions are satisfied and identified the NCM of $h_s$, we now apply Proposition \eqref{th:J} to explicitly write $J_s$ using \eqref{eq:Jsol} and \eqref{eq:hhs-explicit} as
\begin{align}
    J_s(r) & = u_+(r)\hh_s\big(g(r)\big) = \begin{cases}
        u_+(r) h_s\big(g(r)\big)  = f_s(r) & r \leq r_s
        \\
        u_+(r)g(r) \frac{h_s\big(g(r_s)\big)}{g(r_s)} = f_s(r_s)\frac{u_-(r)}{u_-(r_s)} & r > r_s
    \end{cases}. \label{eq:Js-sol}
\end{align}
Next, from \eqref{eq:tau_sol} we can calculate the selling continuation region and optimal selling time as
\begin{align}
    \Cc_s & := \{r : J_s(r) > f_s(r)\} = (r_s,\infty), & \tau^*_s &  := \inf\{  t \geq 0 : R_t \leq r_s\}.
    \label{eq:taus_sol}
\end{align}
In words, the investor's optimal selling strategy is to sell his home the first time the risk-free rate of interest is at or below $r_s$.
Note that the form of $\tau^*_s$ agrees with our previous speculation of the form of the optimal selling strategy in $\eqref{eq:tau-intuitive}$. 

\subsection{Optimal home buying problem}
Having obtained the optimal selling strategy $\tau_s^*$, we now turn our attention to finding the optimal buying strategy $\tau_b^*$.
To begin, we define $h_b$ from $f_b$ in  the same way we define  $h$ from $f$ in \eqref{eq:hdef} by
\begin{align}
    h_b(q) & :=  \frac{f_b\big(g^{-1}(q)\big)}{u_+\big(g^{-1}(q)\big)}, &  q < 0. \label{eq:hb-def}
\end{align}
Using the form of $f_b$ \eqref{eq:fb-def}, the limit expressions \eqref{eq:ulim}, and the explicit form of $J_s$ \eqref{eq:Js-sol}, it is straightforward to confirm that
\begin{align} \lim_{r\to 0^+} \frac{f^+_b(r)}{u_-(r)} & = 0, &
\lim_{r\to \infty} \frac{f^+_b(r)}{u_+(r)} & = 0.
\end{align}
Thus, the limit conditions in Proposition \ref{th:J} are satisfied. Next we need to identify $\hh_b$, the NCM of $h_b$, which is done in Appendix \ref{sec:hhb_cal}. We have from \eqref{eq:hhb-explicit} that
\begin{align}
    \hh_b(q) & = \begin{cases}
        h_b(q_b) & q \leq q_b
        \\
        h_b(q) &  q > q_b
    \end{cases}, & q_b:= g(r_b),
\end{align}
where the buying threshold is the unique positive solution to the following equation \eqref{eq:rb}, which we repeat here for the reader's convenience
\begin{align}
\frac{u_+'(r_b)}{u_+(r_b)} & = \frac{f'_b(r_b)}{f_b(r_b)}. 
\end{align}
Having confirmed that the limit conditions are satisfied and identified the NCM of $h_b$, we now apply Proposition \ref{th:J} to explicitly write $J_b$ using \eqref{eq:Jsol} and \eqref{eq:hhb-explicit} as
\begin{align}
    J_b(r) & = u_+(r)\hh_b\big(g(r)\big)
    =
    \begin{cases}
        u_+(r) h_b\big(g(r_b)\big)  = f_b(r_b)\frac{u_+(r)}{u_+(r_b)} & r \leq r_b
        \\
        u_+(r) h_b\big(g(r)\big) = f_b(r) &  r > r_b
    \end{cases} .
    \label{eq:Jb-sol}
\end{align}
Next, from \eqref{eq:tau_sol} we calculate the buying continuation region and the optimal buying time as
\begin{align}
    \Cc_b & := \{r : J_b(r) > f_b(r)\} = (0,r_b), & \tau^*_b &  := \inf\{  t \geq 0 :R_t \geq r_b\}.
    \label{eq:taub_sol}
\end{align}
In other words, the investor's optimal buying {strategy} is to purchase a {home} the first time the risk-free rate of interest is at or above $r_b$. 
Note again that the form of the optimal buying rule agrees with our speculation of the optimal buying strategy as described in \eqref{eq:tau-intuitive}.

\subsection{{Density} of waiting time}
\label{sec:wait}
The goal of this section is to derive the {densities} and expected values of the optimal selling and buying times $\tau^*_s$ and $\tau^*_b$, which are characterized by \eqref{eq:taus_sol} and \eqref{eq:taub_sol}, respectively.  These quantities are important because, for example, if the expected value of either $\tau^*_b$ or $\tau^*_s$ are on the order of 100s of years, then it would not be practical for an investor to implement the optimal buying and/or selling strategies.
\\[0.5em]
To begin our analysis, let us define the probability density functions of $\tau_b^*$ and $\tau_s^*$.  We have
\begin{align}
p_{\tau_b^*}(t;r)
	&:= \frac{\dd}{\dd t} \Pb( \tau_b^* \leq t | R_0 = r ) , &
r
	&< r_b , \\
p_{\tau_s^*}(t;r)
	&:= \frac{\dd}{\dd t} \Pb( \tau_s^* \leq t | R_0 = r ) , &
r
	&> r_s .
\end{align}
Note {that} we have restricted the definitions of $p_{\tau_b^*}$ and $p_{\tau_s^*}$ to cases in which $r < r_b$ and $r > r_s$ because if $r \geq r_b$ we have trivially that $\tau_b^* = 0$ and if $r \leq r_s$ we have trivially that $\tau_s^* = 0$.
Note also that $p_{\tau_b^*}$ is the density of the first hitting time of $R$ to level $r_b$ from below and $p_{\tau_s^*}$ is the density of the first hitting time of $R$ to level $r_s$ from above.  The first hitting time densities for the CIR process are computed explicitly in \cite[Proposition 1]{linetsky2004computing}, which we present below using the notation of the present paper.
\begin{proposition}
\label{th:taudis}
Suppose that the risk-free rate of interest $(R_t)_{t \geq 0}$ is a CIR process defined in \eqref{eq:CIR} with parameters $(\kappa, \theta, \sig)$ that satisfies Feller's condition. Suppose that the initial interest rate $r$, the buying threshold $r_b$, and the selling threshold $r_s$ are such that $r_s<r< r_b$. Let $\big(k_{b,n}(r_b)\big)_{n \geq 1},\big(k_{s,n}(r_s)\big)_{n \geq 1} $ be the decreasing negative sequences that are all negative roots of the equations
\begin{align}
    M(k_{b,n}(r_b),\beta,\om r_b) &  = 0, & 
    U(k_{s,n}(r_s),\beta,\om r_s) & = 0, &
    \beta & := \frac{2\kappa \theta}{\sig^2}, & \om & := \frac{\beta}{\theta}, 
\end{align}
respectively,and $\big(m_{b,n}(r,r_b)\big)_{n \geq 1},\big(m_{s,n}(r,r_s)\big)_{n \geq 1}$ by
\begin{align} m_{b,n}(r,r_b) & := -\frac{M(k_{b,n}(r_b),\beta,\om r)}{k_{b,n}(r_b) \frac{\d }{\d k}M(k,\beta,\om r_b)|_{k=k_{b,n}(r_b)}},
    \\
    m_{s,n}(r,r_s) &:= -\frac{U(k_{s,n}(r_s),\beta,\om r)}{k_{s,n}(r_s) \frac{\d }{\d k}U(k,\beta,\om r_s)|_{k=k_{s,n}(r_s)}}.
\end{align}
Then the probability density functions of $\tau_b^*$ and $\tau_s^*$ are given by
\begin{align}
    p_{\tau_b^*}(t;r) & = -\kappa\sum_{n=1}^{\infty} m_{b,n}(r,r_b)  k_{b,n}(r_b) \ee^{\kappa k_{b,n}(r_b) t}, \label{eq:pdf-buy}
    \\
		p_{\tau_s^*}(t;r) & = -\kappa\sum_{n=1}^{\infty}
		m_{s,n}(r,r_s)  k_{s,n}(r_s) \ee^{\kappa k_{s,n}(r_s) t}, \label{eq:pdf-sell}
\end{align}
respectively. The uniform convergence of the infinite series \eqref{eq:pdf-buy} and \eqref{eq:pdf-sell} are proven in \cite[Proposition 2]{linetsky2004computing}.
\end{proposition}
From \cite[Equation 19 and 20]{linetsky2004computing} the coefficients $ k_{b,n}(r_b)$ and $m_{b,n}(r,r_b)$  have the following large-$n$ asymptotics
\begin{align}
    k_{b,n}(r_b) & = \Oc(-n^2), & \vert m_{b,n}(r,r_b) \vert & = \Oc(\frac{1}{n}),  \label{eq:buy-bigo}    
\end{align}
and using \cite[Equation 23 and 24]{linetsky2004computing}, the coefficients $k_{s,n}(r_s)$ and $m_{s,n}(r,r_s)$ have the following large-$n$ asymptotics 
\begin{align}
     k_{s,n}(r_s) & = \Oc(-n), & \vert m_{s,n}(r,r_s) \vert & = \Oc(\frac{1}{n}).  \label{eq:sell-bigo}
\end{align}
The large-$n$ asymptotics of the coefficients in \eqref{eq:buy-bigo} and \eqref{eq:sell-bigo} guarantee that the infinite sums in the computation of expectations, which we perform below in \eqref{eq:buy-ave}, \eqref{eq:sell-ave} and \eqref{eq:all-ave}, converge absolutely. Thus, the infinite sums and integrals can be exchanged.
\\[0.5em]
Using Proposition \ref{th:taudis} we can compute the expected length of time the investor will wait prior to buying a home assuming he follows the optimal  buying strategy.  We have
\begin{align}
\Eb\Big(\tau^*_{b} \Big| R_0 = r < r_b \Big)
& = \int_0^{\infty} t \, p_{\tau^*_b}(t;r) \dd t 
	= -\kappa\sum_{n=1}^{\infty}\int_0^{\infty}  m_{b,n}(r,r_b)  k_{b,n}(r_b) t\ee^{\kappa k_{b,n}(r_b) t} \dd t 
	\\
	& = -\frac{1}{\kappa}\sum_{n=1}^{\infty} \frac{m_{b,n}(r,r_b)}{k_{b,n}(r_b)} . \label{eq:buy-ave}
\end{align}
Similarly, the expected length of time the investor will wait prior to selling a home after buying it assuming he follows the optimal buying and selling strategies is
\begin{align}
\Eb \Big(\tau^*_s \Big| R_0 = r_b \Big) 
	& = \int_0^{\infty} t \, p_{\tau_s^*}(t;r_b) \dd t 
	= -\kappa\sum_{n=1}^{\infty}\int_0^{\infty}  m_{s,n}(r_b,r_s)  k_{s,n}(r_s) t\ee^{\kappa k_{s,n}(r_s) t} \dd t
	\\
	&
	= -\frac{1}{\kappa}\sum_{n=1}^{\infty} \frac{m_{s,n}(r_b, r_s)}{k_{s,n}(r_s)}. \label{eq:sell-ave}
\end{align}
\\[0.5em]
Lastly, we are interested to know the probability density function of $\tau_b^* + \tau_s^*$ the total time the investor waits to buy and then sell a home, assuming he follows the optimal buying and selling strategies.  The probability density function of $\tau_b^*+\tau_s^*$, given by
\begin{align}
    p_{\tau_b^*+\tau_s^*}(t;r):= \frac{\dd}{\dd t} \Pb( \tau_b^*+\tau_s^* \leq t | R_0 = r),
\end{align}
can be calculated as a convolution of the two probability densities \eqref{eq:pdf-buy} and \eqref{eq:pdf-sell}.  We have
\begin{align}
p_{\tau_b^* + \tau_s^*}(t;r)
	&=	\int_0^t p_{\tau_b^*}(t';r) p_{\tau_s^*}(t-t';r_b) \dd t'
	\\
	& = \kappa^2\int_{0}^t \left(\sum_{i=1}^{\infty}m_{b,i}(r,r_b) k_{b,i}(r_b) \ee^{\kappa k_{b,i}(r_b) t'} \sum_{j=1}^{\infty}m_{s,j}(r_b,r_s) k_{s,j}(r_s) \ee^{\kappa k_{s,j}(r_s) (t-t')} \right) \dd t'
	\\
	& = \kappa^2 \sum_{i,j=1}^{\infty} \ee^{\kappa k_{s,j}(r_s)t} \int_{0}^t m_{b,i}(r,r_b) k_{b,i}(r_b) m_{s,j}(r_b,r_s) k_{s,j}(r_s) \ee^{\kappa(k_{b,i}(r_b)-k_{s,j}(r_s)) t'} \dd t' 
	\\
	& =  \kappa^2 \sum_{i,j=1}^{\infty} m_{b,i}(r,r_b) k_{b,i}(r_b) m_{s,j}(r_b,r_s) k_{s,j}(r_s) \frac{\ee^{\kappa k_{b,i}(r_b)t}-\ee^{\kappa k_{s,j}(r_s) t}}{\kappa k_{b,i}(r_b) -\kappa k_{s,j}(r_s)}. \label{eq:pdf-both}
\end{align}
The expectation of $\tau^*_b+\tau^*_s$ is simply the sum of expectations of $\tau^*_b$ and $\tau^*_s$, which are given {in \eqref{eq:buy-ave} and \eqref{eq:sell-ave}.}

\subsection{Numerical Example}
\label{sec:numerical}
Throughout this section we fix the following parameter values
\begin{align}
\left. \begin{aligned}
 \kappa & = 0.9, & \theta & = \frac{0.08}{0.9}, & \sig & = \sqrt{0.033},
\\
 \gam & = 0.4, & \chi & = 0.6, & r & = 0.08, 
\\
 C & = \$ 100,000 , & 
\rho & = 0.01, & T & =  30 \text{ (years)}, 
\\
\del_b & = \del_s = 0.06 , & K_b & = K_s = \$5000. 
\end{aligned} \right\} \label{eq:ex_param}
\end{align}
The parameters specific to the CIR model $(\kappa,\theta,\sig)$ and initial risk-free rate of interest $r$ were taken from \cite[Example 10.3.2.2]{filipovic2009term}. Note that the parameters $(\kappa,\theta,\sigma)$ defined in \eqref{eq:ex_param} satisfy the {Feller} condition $(2\kappa\theta > \sig^2)$. 
The duration of the loan $(T = 30 \text{ years})$ is standard for a fixed-rate mortgage in the United States.
The fixed and proportional transaction costs are also typical for a US-based mortgage.
\\[0.5em]
In Figure \ref{fig:J_sol}, we plot $J_s$ and $J_b$  using the expressions of the selling and buying value function \eqref{eq:Js-sol}, and \eqref{eq:Jb-sol}. Note that $J_b(r)$ is an increasing function of $r$ because the short term home price is inversely related to interest rate. Likewise, the function $J_s(r)$ is a decreasing function of $r$. Next, using \eqref{eq:rs} and \eqref{eq:rb}, we obtain numerically the selling and buying threshold $r_s \approx 0.026$ and $r_b \approx 0.167$. We plot the probability density function of $\tau_b^*$, the length of time the investor waits before buying, the probability density function of $\tau_s^*$, the length of time the investor holds a home before selling, and the probability density function of $\tau_b^* + \tau_s^*$, the sum of both waiting times in Figure \ref{fig:pdfwait}.
Finally, in order to compute the expected length of time the investor waits before buying a home and the expected length of time the investor holds a home before selling it, assuming he follows the optimal strategies, we truncate the infinite sums in \eqref{eq:buy-ave} and \eqref{eq:sell-ave} at 100 terms and 
obtain
\begin{align}
\Eb\Big(\tau^*_{b} \Big| R_0 = r < r_b \Big)
	&\approx 8.108 , &
\Eb\Big(\tau^*_{s} \Big| R_0 = r_b  \Big)
	&\approx 11.301, & \Eb\Big(\tau^*_b+\tau^*_s \Big| R_0 = r  \Big)
	&\approx 19.409. \label{eq:all-ave}
\end{align}
These expectations are shown as vertical bars in their respective graphs in Figure \ref{fig:pdfwait}.

\section{Conclusion}
\label{sec:conclusion}
In this paper, we have provided an expression for home prices as a function of risk-free rate of interest and its time integral, and the rate of wage growth. In this setting, we have considered an investor who wishes to maximize the discounted expected profit from buying a home and selling it at a later time. Using the expression of home prices, we have defined the optimal home buying and selling problems as a nested optimal stopping problem, for which its value function and optimal stopping rule can be characterized using a nonnegative concave majorant approach.  When the risk-free rate of interest is modeled by a CIR process, we have provided an explicit characterization of the optimal buying and selling times. Additionally, in the case of CIR interest rates, we have analyzed the expected time the investor waits before buying as well as the expected time the investor waits  before selling a home,  assuming he follows the optimal buying and selling strategies. In future work, we plan to extend our results to include a scenario where the investor repeatedly buys and sells homes.

\clearpage
\appendix 
\section{Expressions for $u_+$ and $u_-$}
\label{sec:solveu}
We solve \eqref{eq:uspecific} following \cite{carmona2007investment}. Consider the substitution $u(r):= \ee^{-\nu r}v(r)$ where 
\begin{align}
    \nu & := \frac{\xi-\kappa}{\sig^2}, & \xi & := \sqrt{\kappa^2 + 2 \sig^2(\chi-\gam)},
\end{align}
then \eqref{eq:uspecific} simplifies to
\begin{align}
    r v''(r) + \Big( \frac{2\kappa \theta}{\sig^2}-\frac{2\xi}{\sig^2}r \Big) v'(r) + 2 \frac{\kappa \theta\nu}{\sig^2}v(r) & = 0. \label{eq:v-ode}
\end{align}
Performing the change of variable $v(r): = w(\zeta r)$ where $\zeta:= \frac{2\xi}{\sig^2}$ in \eqref{eq:v-ode} we obtain that $w(r)$ satisfies
\begin{align}
    r w''(r) + ( \frac{2\kappa \theta}{\sig^2}-r)w'(r) - \frac{\kappa \theta}{\sig^2}(1-\frac{\kappa}{\xi})w(r) = 0,
\end{align}
which, with the shorthand $\alpha: = \frac{\kappa}{ \theta}{\sig^2}(1-\frac{\kappa}{\xi}), \beta:= \frac{2 \kappa \theta   }{\sigma ^2}$, can be written as
\begin{align}
        rw''(r)+(\beta-r)w'(r)-\alpha w(r) = 0. \label{eq:sKummer}
\end{align}
Equation \eqref{eq:sKummer} is commonly known as \emph{Kummer's Equation} which has two independent solutions $w = (w_+,w_-)$ where
\begin{align}
    w_+(r) & = M(\alpha,\beta,r)= M\Big(\frac{\kappa \theta}{\sig^2}(1-\frac{\kappa}{\xi}),\frac{2 \kappa \theta   }{\sigma ^2},r\Big),
    \\
    w_-(r) & = U(\alpha,\beta,r) = U\Big(\frac{\kappa \theta}{\sig^2}(1-\frac{\kappa}{\xi}),\frac{2 \kappa\theta   }{\sigma ^2},r\Big),
\end{align}
and where $M$ and $U$ are \textit{Confluent Hypergeometric Function} (CHF) of the first kind and second kind, defined by
\begin{align}
    M(\alpha,\beta,r) & = \sum_{n=0}^{\infty} \frac{\alpha(\alpha+1)\ldots (\alpha+n)}{\beta(\beta+1)\ldots (\beta+n)} \frac{r^n}{n!},
    \\
    U(\alpha,\beta,r) & = \frac{\Gam (1-\beta)}{\Gam (\alpha+1-\beta)}M(\alpha,\beta,r)+\frac{\Gam (\beta-1)}{\Gam (\alpha)}r^{1-\beta}M(\alpha+1-\beta,2-\beta,r),
    \\
    & = \frac{1}{\Gam(\alpha)}\int_{0}^{\infty} \dd t \, \ee^{-rt} t^{\alpha-1} (1+t)^{\beta-\alpha-1},
    \label{eq:MU_def}
\end{align}
and $\Gam$ is the \emph{{Euler} gamma function}. Note that since $\chi > \gam$, then the parameters $\alpha,\beta, \nu,\xi$ are all positive. Substitute back $w,v$ into $u$ we obtain $u = \big(u_+(r),u_-(r)\big) = \big(\ee^{-\nu r} w_+(\zeta r), \ee^{-\nu r} w_-(\zeta r)\big)$ which is the form of \eqref{eq:u12}. It is clear that since each parameter in the argument of CHF is positive, $u_+$ and $u_-$ are positive. Next we will show that $u_+$ and $u_-$ are strictly increasing and decreasing, respectively. First we establish some basic {properties of CHFs, which are well known.}
\begin{lemma} We have that the derivatives of CHF of the first and second kind are
\begin{align}
    \frac{\dd}{\dd r} M(\alpha,\beta,\zeta r) & = \frac{\alpha\zeta}{\beta} M(\alpha+1,\beta+1,\zeta r), &
    \frac{\dd}{\dd r} U(\alpha,\beta,\zeta r) & = -\alpha\zeta U(\alpha+1,\beta+1,\zeta r). \label{eq:MU_diff}
\end{align}
\end{lemma}
\begin{proof}
We can show \eqref{eq:MU_diff} by noting that
\begin{align}
    \frac{\dd}{\dd r} M(\alpha,\beta,\zeta r) & = \sum_{n=0}^{\infty} \frac{\dd}{\dd r} \frac{\alpha(\alpha+1)\ldots (\alpha+n)}{\beta(\beta+1)\ldots (\beta+n)} \frac{\zeta^n r^n}{n!}
    \\
    & = \sum_{n=1}^{\infty}  \frac{\alpha(\alpha+1)\ldots (\alpha+n)}{\beta(\beta+1)\ldots (\beta+n)} \frac{\zeta^n r^{n-1}}{(n-1)!}
    \\
    & = \frac{\alpha\zeta}{\beta} \sum_{n=0}^{\infty}  \frac{(\alpha+1)\ldots (\alpha+1+n)}{(\beta+1)\ldots (\beta+1+n)} \frac{\zeta^n r^{n}}{n!} = \frac{\alpha\zeta}{\beta}M(\alpha+1,\beta+1,r).
\end{align}
Using the relationship between $M$ and $U$ in \eqref{eq:MU_def} and the derivative formula of $M$, we perform similar calculation to obtain the derivative formula for $U$.
\end{proof}
\begin{lemma} {The function} $u_+$ is increasing and the function $u_-$ is decreasing.
\end{lemma}
\begin{proof}
Note that using \eqref{eq:MU_diff} we obtain
\begin{align}
    \frac{\dd u_+(r)}{\dd r} & = \frac{\dd}{\dd r}\Big(\ee^{-\nu r}  M(\alpha,\beta,\zeta r)\Big) = \ee^{-\nu r}\Big(-\nu M(\alpha,\beta,\zeta r) + \frac{\alpha \zeta}{\beta}M(\alpha+1,\beta+1,\zeta r) \Big)
    \\
    & = -\nu\ee^{-\nu r}\Big( M(\alpha,\beta,\zeta r)- M(\alpha+1,\beta+1,\zeta r) \Big)
    \\
    & = -\nu\ee^{-\nu r} \sum_{n=0}^{\infty}\Big( \frac{\alpha(\alpha+1)\ldots (\alpha+n)}{\beta(\beta+1)\ldots (\beta+n)}- \frac{(\alpha+1)\ldots (\alpha+1+n)}{(\beta+1)\ldots (\beta+1+n)} \Big) \frac{(\zeta r)^n}{n!}
    \\
    & = -\nu\ee^{-\nu r} (\alpha-\beta) \sum_{n=0}^{\infty}\frac{(\alpha+1)\ldots (\alpha+n)}{(\beta+1)\ldots (\beta+n)}\Big( \frac{1+n}{\beta(\beta+1+n)} \Big) \frac{(\zeta r)^n}{n!}
    \\
    & = \ee^{-\nu r}\frac{\alpha \zeta}{\beta} \frac{\kappa\theta}{\sig^2}(1+\frac{\kappa}{\xi}) \sum_{n=0}^{\infty}\frac{(\alpha+1)\ldots (\alpha+n)}{(\beta+1)\ldots (\beta+n)}\Big( \frac{1+n}{\beta(\beta+1+n)} \Big) \frac{(\zeta r)^n}{n!} > 0. \label{eq:u+diff}
\end{align}
This means that $u_+$ is increasing. Note that using \eqref{eq:MU_diff} we obtain
\begin{align}
    \frac{\dd u_-(r)}{\dd r} = \frac{\dd}{\dd r}\Big(\ee^{-\nu r}  U(\alpha,\beta,\zeta r)\Big) & = \ee^{-\nu r}\Big(-\nu U(\alpha,\beta,\zeta r) - \alpha \zeta U(\alpha+1,\beta+1,\zeta r) \Big) < 0. \label{eq:u-diff}
\end{align}
Since $\nu,\alpha,\zeta > 0$,  $U(\alpha,\beta,\zeta r)>0$ and  $U(\alpha+1,\beta+1,\zeta r) > 0$ which is clear from the integral representation of $U$ in $\eqref{eq:MU_def}$, $u_-$ is decreasing.
\end{proof}
\section{Expressions for $\hh_s$ and $\hh_b$}
\label{sec:hcal}
    \subsection{Expressions for $\hh_s$}
    \label{sec:hhs_cal}
    We plot $h_s$ defined in \eqref{eq:hs-def} using parameters in \eqref{eq:ex_param} in Figure \ref{fig:hs} which shows that $h_s(q)$ is concave for $q \in (-\infty, q_s^*)$ and becomes convex for $q \in (q_s^*,0)$ where $q^*_s$ is the inflection point of $h_s$. {The value of} $q^*_s$ can be obtained numerically by using the expression for $h''_s(q)$ from \eqref{eq:h-diff2} and setting $h''_s(q^*_s) = 0$. Suppose that the tangent line of $h_s$ passing through $0$ intersects $h_s$ at $\big(q_s,h_s(q_s)\big)$ for some point $q_s$. We can solve for $q_s$ from \eqref{eq:qs}
\begin{align}
    \frac{h_s(q_s)}{q_s} = h_s'(q_s). \label{eq:qs}
\end{align}
Substituting $q_s = g(r_s)$ and using \eqref{eq:h-diff1} we obtain
\begin{align}
    \frac{h_s(q_s)}{q_s} & = -\frac{f_s(r_s)}{u_-(r_s)}, & h_s'(q_s) & = \frac{f'_s(r_s)u_+(r_s) -f_s(r_s)u_+'(r_s)}{u_+'(r_s) u_-(r_s)-u_+(r_s) u_-'(r_s)}. \label{eq:qs-sim}
\end{align}
Using \eqref{eq:qs-sim} we can rewrite \eqref{eq:qs} in terms of $r_s$ by
\begin{align}
   \frac{u_-'(r_s)}{u_-(r_s)} & = \frac{f'_s(r_s)}{f_s(r_s)}. \label{eq:rs}
\end{align}
{Equation} \eqref{eq:rs} can be solved numerically to obtain the selling threshold $r_s$. From Figure \ref{fig:hs} we can see that since $q_s < q^*_s$, the NCM of $h_s$ is the $h_s$ itself on $(-\infty,q_s)$ and on $(q_s,0)$ it is the tangent line to $h_s$ passing through $0$. With all the information we can explicitly write $\hh_s$ as
\begin{align}
    \hh_s(q) = \begin{cases}
    h_s(q)  & q \leq q_s
    \\
    q\frac{h_s(q_s)}{q_s} & q > q_s \label{eq:hhs-explicit}
    \end{cases}.
\end{align}
\subsection{Expressions for $\hh_b$}
 \label{sec:hhb_cal}

We plot $h_b$ defined in \eqref{eq:hb-def} using the parameters in \eqref{eq:ex_param} in Figure \ref{fig:hb}. Suppose that $q_b$ and $q^*_b$ is the critical point and inflection point of $h_b$, respectively, which can be numerically solved from finding the roots of the first and second derivative of $h_b$ in \eqref{eq:h-diff1} and \eqref{eq:h-diff2}. 
From Figure \ref{fig:hb}, we can see that the NCM of $h_b$ is the horizontal line (note that the NCM of an increasing function is the horizontal line of the maximum of that function) with value $q_b$ on $(-\infty,q_b)$. Since $q^*_b < q_b$, on $(q_b,0)$ the graph of $h_b$ is decreasing, nonnegative and concave, so clearly the NCM of this part of the function is $h_b$ itself. With all the information we can write $\hh_b$ as
\begin{align}
    \hh_b(q) = \begin{cases}
        h_b(q_b) & q \leq q_b
        \\
        h_b(q) &  q > q_b \label{eq:hhb-explicit}
    \end{cases},
\end{align}
where $q_b$ can be solved by setting \eqref{eq:h-diff1} to zero which is equivalent to
\begin{align}
    \frac{u'_+(r_b)}{u_+(r_b)} = \frac{f'_b(r_b)}{f_b(r_b)}. \label{eq:rb}
\end{align}
We solve for the buying threshold $r_b$ using \eqref{eq:rb} and set $q_b = g(r_b)$ to obtain $q_b$.
\section{Constant discount rate}
\label{sec:constant-discount}
Suppose that, for a payment $P_t$ received at time $t$ the investor assigns a present value of $\Eb\big(\ee^{-\chi t}P_t\big)$ instead of $\Eb\big(\ee^{-\chi\int_0^t R_s \dd s}P_t\big)$ as in Section \ref{sec:optimal}. Revising the processes in Section \ref{sec:optimal} to obtain the optimal buying and selling problems, we can see that all steps applied in Section \ref{sec:optimal} can also be applied in this setting, except that the process $\Lam$ is changed to a modified process reflecting the change in integral discounting rate to constant discounting rate. The modified process $\Lamt =(\Lamt_t)_{t\geq0}$ is defined by 
\begin{align}
\Lamt_t 
	&:= \chi t -\gam\int_{0}^{t} \, R_s \dd s.
\end{align}
Having the form of the modified process, we can define a \emph{modified buying value function and selling function}, $\Jt_b$ and $\Jt_s$ as
\begin{align}
    \Jt_b(r) &:=	\sup_{\tau_b \in \Tc}\Eb \Big[ \ee^{-\Lamt_{\tau_b}} f_b(R_{\tau_b}) \Big| R_0 = r \Big],
     &
    \Jt_s(r)
	&:=	\sup_{\tau_s \in \Tc}\Eb \Big[ \ee^{-\Lamt_{\tau_s}} f_s(R_{\tau_s}) \Big| R_0 = r \Big],
     \label{eq:Jt-def}
\end{align}
and \emph{modified optimal selling and buying strategies} $\taut^*_s$ and $\taut^*_b$ are defined in the same way as \eqref{eq:taus_sol} and \eqref{eq:taub_sol}.  The functions $f_s$ and $f_b$ remain unchanged and are given by \eqref{eq:fs-def} and \eqref{eq:fb-def}, respectively. In this setting, the ODE \eqref{eq:uode} becomes
\begin{align}
    \Big(\Ac-\chi + \gam r \Big) \ut(r) & = 0. \label{eq:uspecificnew}
\end{align}
Note that for a given generator $\Ac$ and constant $\gam$, positive strictly increasing and positive strictly decreasing solutions of \eqref{eq:uspecificnew} will only exist if $\chi$ is larger than some threshold value.

\subsection{CIR Interest Rate}
\label{sec:constant-discount-cir}
We now focus on the case in which the interest rate $R$ is modeled by CIR process. In this setting, \eqref{eq:uspecificnew} is given by
\begin{align}
    \Big(\kappa (\theta - r ) \d_r + \tfrac{1}{2} \sig^2 r \d_r^2-\chi + \gam r \Big) \ut(r) & = 0. 
\end{align}
The solution $\ut = (\ut_+,\ut_-)$, where $\ut_+$ is positive strictly increasing and $\ut_-$ is positive strictly decreasing are of the form \eqref{eq:u12new}
\begin{align}
        \ut_+(r) & = \ee^{-\nu r}  M(\alpha,\beta,\zeta r),
        &
        \ut_-(r) & = \ee^{-\nu r}  U(\alpha,\beta,\zeta r),\label{eq:u12new}
\end{align}
where the parameters $(\alpha, \beta, \xi, \zeta, \nu)$ are defined as
\begin{align}
    \alpha & :=\frac{\kappa \theta  }{{\sigma ^2}} \left(1-\frac{\kappa-\sig^2 \chi/\kappa\theta }{\xi}\right),
    &
    \beta & := \frac{2\kappa \theta   }{\sigma ^2}, & \xi & := \sqrt{\kappa^2 -2\gam\sig^2},
    \\
    \zeta & := \frac{2 \xi}{\sigma ^2}, & \nu &  = \frac{\xi-\kappa}{\sig^2}. \label{eq:u12new_param}
\end{align}
Note that while the form of the functions \eqref{eq:u12new} is the same as \eqref{eq:u12}, the value of the constants has been modified, reflecting the change from stochastic to constant discounting. 
\\[0.5em]
We would like to find a sufficient condition for positive strictly increasing solutions $\ut_+$ and positive strictly decreasing solutions $\ut_-$ to exist. To this end, we note that the derivative of $\ut_+$ can be written using \eqref{eq:u+diff} as
\begin{align}
    \frac{\dd u_+(r)}{\dd r} = -\nu\ee^{-\nu r} (\alpha-\beta) \sum_{n=0}^{\infty}\frac{(\alpha+1)\ldots (\alpha+n)}{(\beta+1)\ldots (\beta+n)}\Big( \frac{1+n}{\beta(\beta+1+n)} \Big) \frac{(\zeta r)^n}{n!}.
\end{align}
As  $\nu < 0, \zeta > 0$, the derivative is guaranteed to be positive if $\alpha > \beta > 0$, which is equivalent to
\begin{align}
    \chi > \frac{\beta}{2}\left( \kappa-\sqrt{\kappa^2-2\gam \sig^2} \right).
\end{align}
Thus, if $\chi$ is large enough, $\ut_+$ will be strictly increasing.  Performing similar analysis using the derivative of $u_-$ in \eqref{eq:u-diff}, we can see that this condition is also sufficient to guarantee $\ut_-$ is decreasing.

\bibliography{bibliography}

\begin{thebibliography}{15}
\providecommand{\natexlab}[1]{#1}
\providecommand{\url}[1]{\texttt{#1}}
\expandafter\ifx\csname urlstyle\endcsname\relax
  \providecommand{\doi}[1]{doi: #1}\else
  \providecommand{\doi}{doi: \begingroup \urlstyle{rm}\Url}\fi

\bibitem[mor(2022)]{mortgage_rate_2022}
S\&p dow jones indices llc, s\&p/case-shiller u.s. national home price index
  [csushpinsa], Jan 2022.
\newblock URL \url{https://fred.stlouisfed.org/series/CSUSHPINSA}.

\bibitem[pri(2022)]{price_index_2022}
Freddie mac, 30-year fixed rate mortgage average in the united states
  [mortgage30us], Jan 2022.
\newblock URL \url{https://fred.stlouisfed.org/series/MORTGAGE30US}.

\bibitem[Albrecht et~al.(2016)Albrecht, Gautier, and
  Vroman]{albrecht2016directed}
J.~Albrecht, P.~A. Gautier, and S.~Vroman.
\newblock Directed search in the housing market.
\newblock \emph{Review of Economic Dynamics}, 19:\penalty0 218--231, 2016.

\bibitem[Anglin(2004)]{anglin2004long}
P.~M. Anglin.
\newblock How long does it take to buy one house and sell another?
\newblock \emph{Journal of Housing Economics}, 13\penalty0 (2):\penalty0
  87--100, 2004.

\bibitem[Brown et~al.(2013)Brown, McGreal, and Adair]{brown2013role}
L.~Brown, S.~McGreal, and A.~Adair.
\newblock The role of bidding in determining sales price for residential
  property.
\newblock \emph{Journal of Housing Research}, 22\penalty0 (1):\penalty0 39--57,
  2013.

\bibitem[Bruss and Ferguson(1997)]{bruss1997multiple}
F.~T. Bruss and T.~S. Ferguson.
\newblock Multiple buying or selling with vector offers.
\newblock \emph{Journal of Applied Probability}, 34\penalty0 (4):\penalty0
  959--973, 1997.

\bibitem[Carmona and Le{\'o}n(2007)]{carmona2007investment}
J.~Carmona and A.~Le{\'o}n.
\newblock Investment option under cir interest rates.
\newblock \emph{Finance Research Letters}, 4\penalty0 (4):\penalty0 242--253,
  2007.

\bibitem[Dayanik(2008)]{dayanik2008optimal}
S.~Dayanik.
\newblock Optimal stopping of linear diffusions with random discounting.
\newblock \emph{Mathematics of Operations Research}, 33\penalty0 (3):\penalty0
  645--661, 2008.

\bibitem[Egozcue et~al.(2013)Egozcue, Garc{\'\i}a, and
  Zitikis]{egozcue2013optimal}
M.~Egozcue, L.~F. Garc{\'\i}a, and R.~Zitikis.
\newblock An optimal strategy for maximizing the expected real-estate selling
  price: accept or reject an offer?
\newblock \emph{Journal of Statistical Theory and Practice}, 7\penalty0
  (3):\penalty0 596--609, 2013.

\bibitem[Filipovic(2009)]{filipovic2009term}
D.~Filipovic.
\newblock \emph{Term-Structure Models. A Graduate Course.}
\newblock Springer, 2009.

\bibitem[Leung and Tse(2017)]{leung2017flipping}
C.~K.~Y. Leung and C.-Y. Tse.
\newblock Flipping in the housing market.
\newblock \emph{Journal of Economic Dynamics and Control}, 76:\penalty0
  232--263, 2017.

\bibitem[Leung and Li(2015)]{leung2015optimal}
T.~Leung and X.~Li.
\newblock Optimal mean reversion trading with transaction costs and stop-loss
  exit.
\newblock \emph{International Journal of Theoretical and Applied Finance},
  18\penalty0 (03):\penalty0 1550020, 2015.

\bibitem[Leung et~al.(2014)Leung, Li, and Wang]{leung2014optimal}
T.~Leung, X.~Li, and Z.~Wang.
\newblock Optimal starting--stopping and switching of a cir process with fixed
  costs.
\newblock \emph{Risk and Decision Analysis}, 5\penalty0 (2-3):\penalty0
  149--161, 2014.

\bibitem[Linetsky(2004)]{linetsky2004computing}
V.~Linetsky.
\newblock Computing hitting time densities for cir and ou diffusions:
  Applications to mean-reverting models.
\newblock \emph{Journal of Computational Finance}, 7:\penalty0 1--22, 2004.

\bibitem[Shreve et~al.(1994)Shreve, Soner, et~al.]{shreve1994optimal}
S.~Shreve, H.~Soner, et~al.
\newblock Optimal investment and consumption with transaction costs.
\newblock \emph{The Annals of Applied Probability}, 4\penalty0 (3):\penalty0
  609--692, 1994.

\end{thebibliography}

\clearpage

\begin{figure}
\centering
\begin{tabular}{cc}
\includegraphics[width=0.5\textwidth]{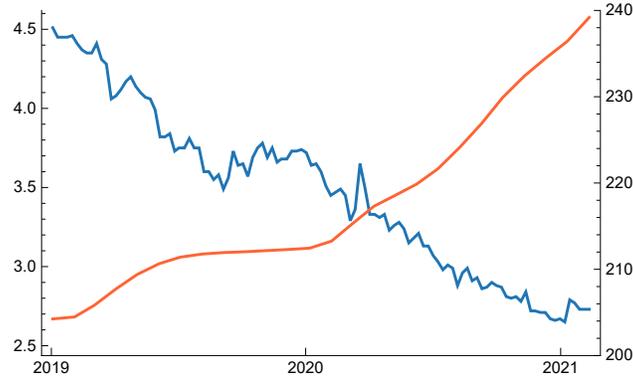}
\\ [1em]
\end{tabular}
\caption{We plot, from January 2019 to January 2021, the Monthly S\&P/Case-Shiller U.S. National Home Price Index \cite{price_index_2022} in orange with the scale on the right vertical axis, and Weekly 30-Year Fixed Rate Mortgage Average in the United States \cite{mortgage_rate_2022} in blue with the scale on the left vertical axis. Note that decreasing the federal mortgage rate has an effect of increasing the home price index during this short-term period.}
\label{fig:house-rate-price}
\end{figure}


\begin{figure}
\centering
\begin{tabular}{cc}
\includegraphics[width=0.5\textwidth]{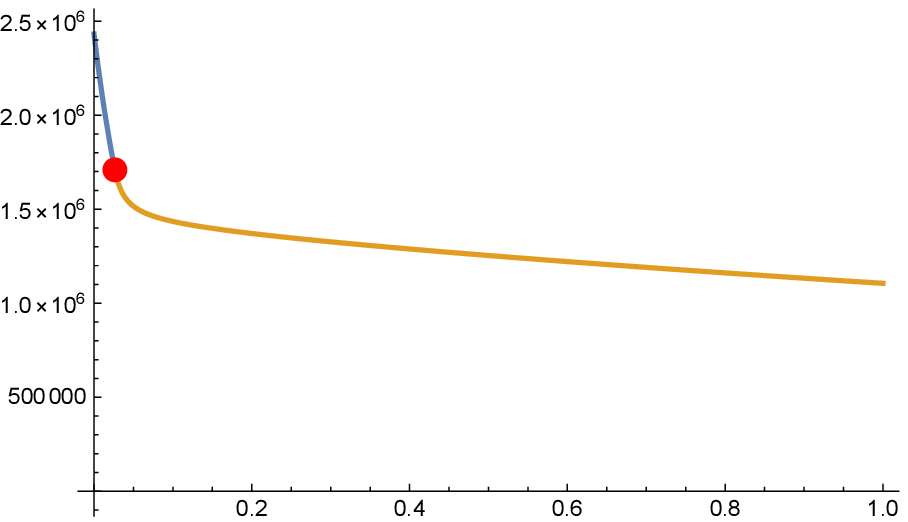}&
\includegraphics[width=0.5\textwidth]{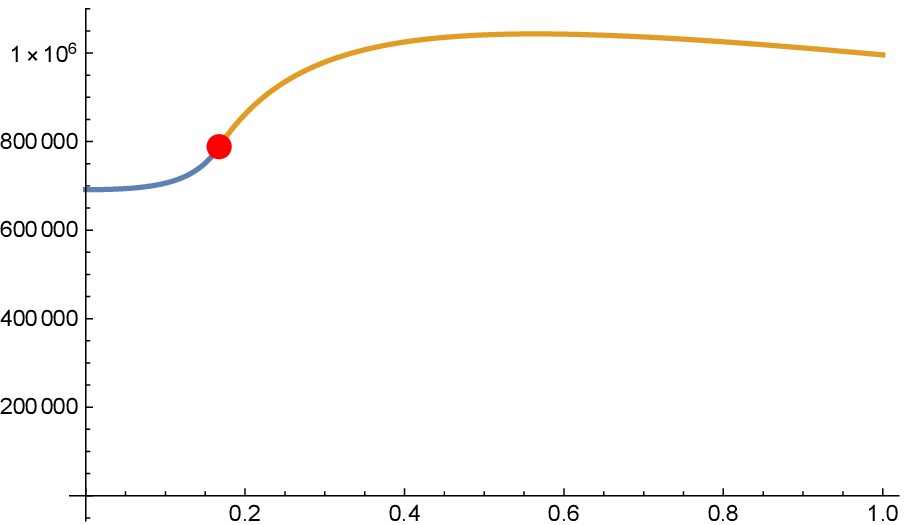}\\
$J_s(r), 0 < r < 1,$ & 
$J_b(r), 0 < r < 1.$ \\ [1em]
\end{tabular}
\caption{We plot the optimal selling and buying functions $J_s(r)$ and $J_b(r)$ when the interest rate is modeled by CIR process using \eqref{eq:Js-sol} and \eqref{eq:Jb-sol} as a function of risk-free rate of interest $r$ where $0 < r < 1$ using parameters defined in \eqref{eq:ex_param}. We numerically solved for the selling and buying threshold using \eqref{eq:rs} and \eqref{eq:rb} to obtain $r_s \approx 0.026$ and $r_b \approx 0.167$ which is shown as red points in each respective graph.}
\label{fig:J_sol}
\end{figure}
\begin{figure}
\centering
\begin{tabular}{cc}
\includegraphics[width=0.45\textwidth]{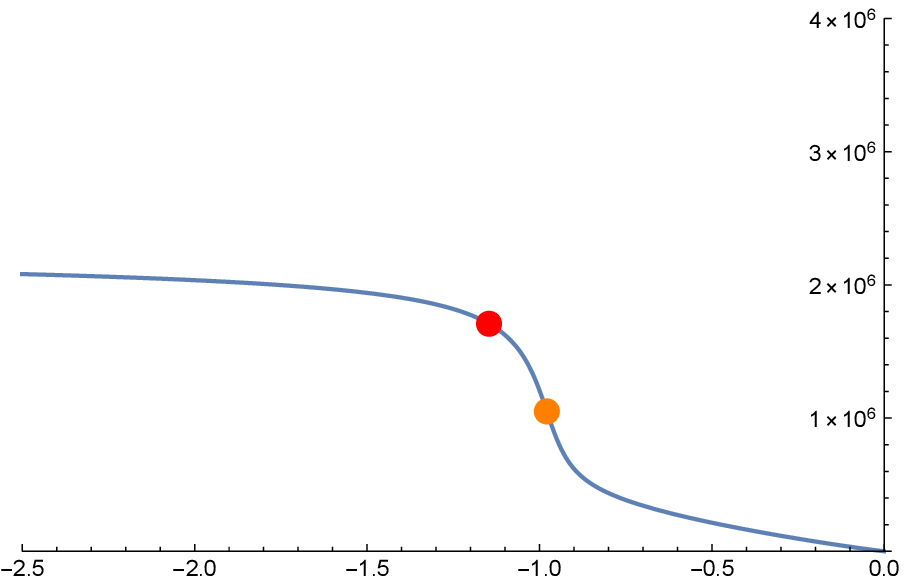}&
\includegraphics[width=0.45\textwidth]{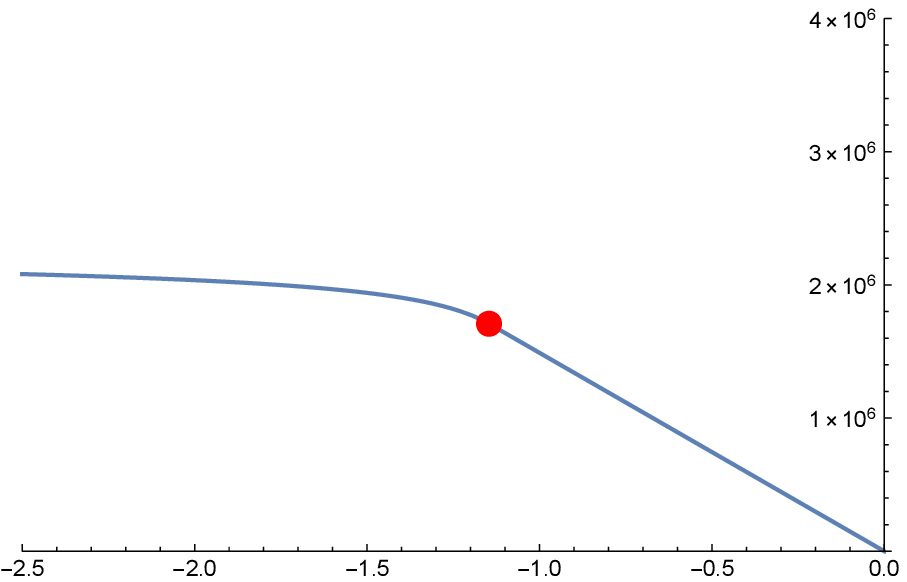}\\
$h_s(q), -2.5 < q < 0,$ & $\hh_s(q), -2.5 < q <0.$\\ [1em]
\end{tabular}
\caption{We plot the graphs of the functions $h_s(q)$ defined in \eqref{eq:hs-def} and its NCM $\hh_s(q)$ defined in \eqref{eq:hhs-explicit} using parameter \eqref{eq:ex_param} for $-2.5 < q < 0$. The red point shows the point $q_s$ which is numerically solved from \eqref{eq:qs}, and the orange point shows the inflection point $q^*_s$ of $h_s$ which is numerically solved using the expression of $h''_s$ in \eqref{eq:h-diff2}.}
\label{fig:hs}
\end{figure}
\begin{figure}
\centering
\begin{tabular}{cc}
\includegraphics[width=0.45\textwidth]{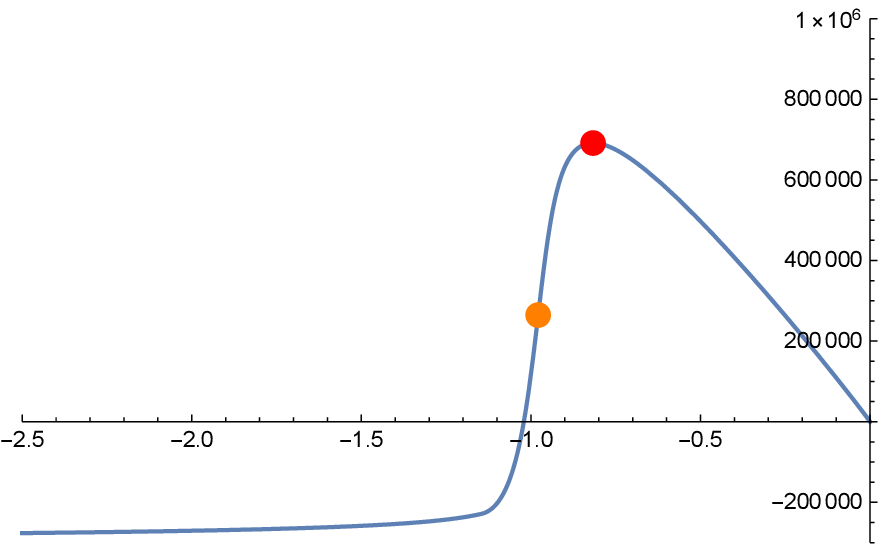}&
\includegraphics[width=0.45\textwidth]{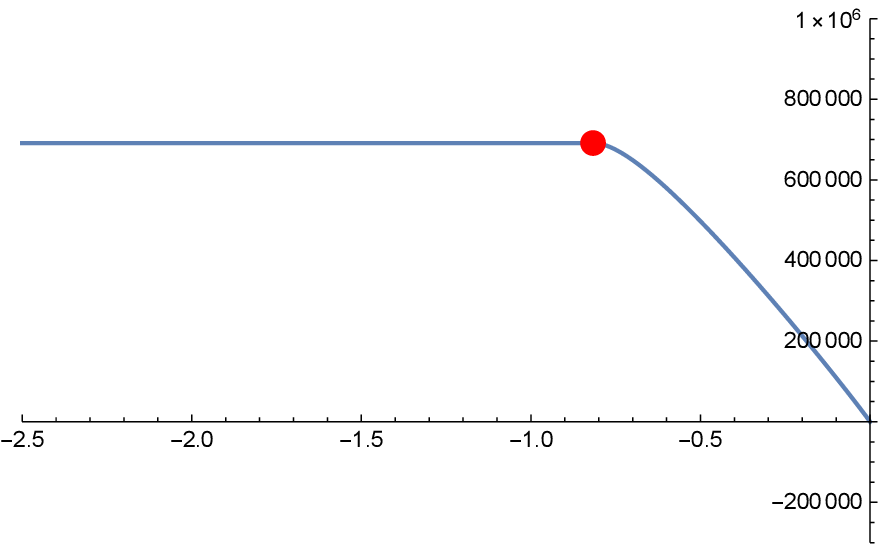}\\
$h_b(q), -2.5 < q < 0,$ & $\hh_b(q), -2.5 < q <0.$ \\ [1em]
\end{tabular}
\caption{We plot the graphs of the functions $h_b(q)$ defined in  \eqref{eq:hb-def} and its NCM $\hh_b(q)$ defined in \eqref{eq:hhb-explicit} using parameter \eqref{eq:ex_param} for $-2.5 < q < 0$. The red point shows the critical point $q_b$ of $h_b$ which is solved numerically using the expression $h'_b$ in \eqref{eq:h-diff1} and orange point shows the inflection point $q^*_b$ of $h_b$ which is numerically solved using the expression of $h''_b$ in \eqref{eq:h-diff2}.}
\label{fig:hb}
\end{figure}

\begin{figure}
\centering
\begin{tabular}{cc}
\includegraphics[width=0.45\textwidth]{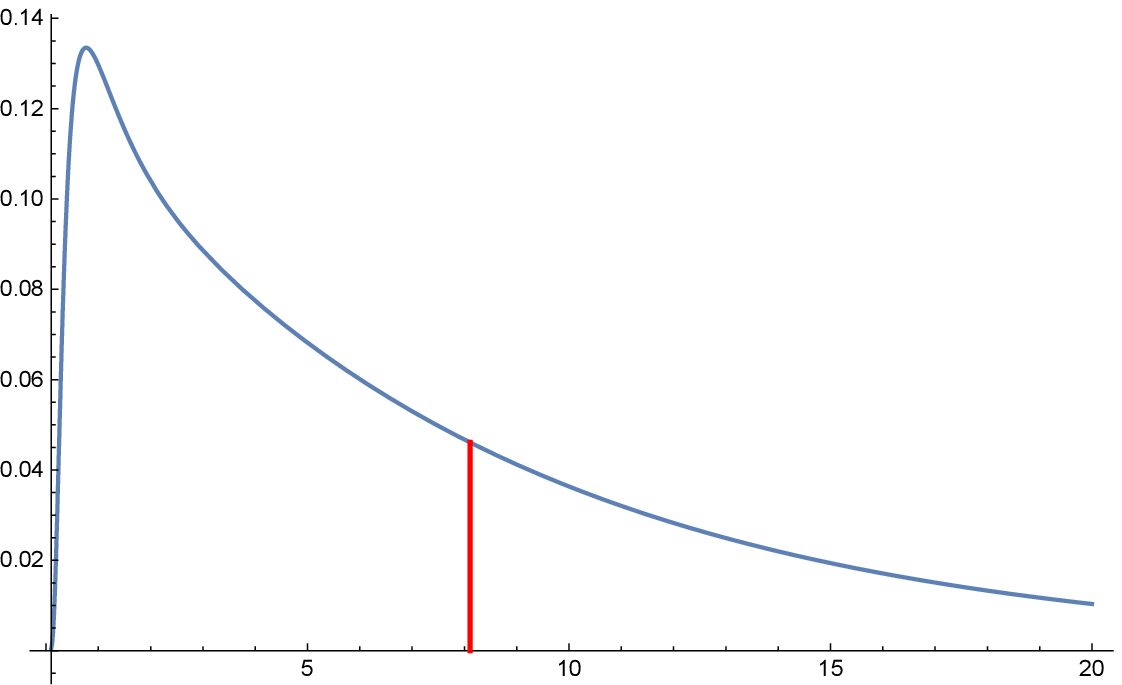}&
\includegraphics[width=0.45\textwidth]{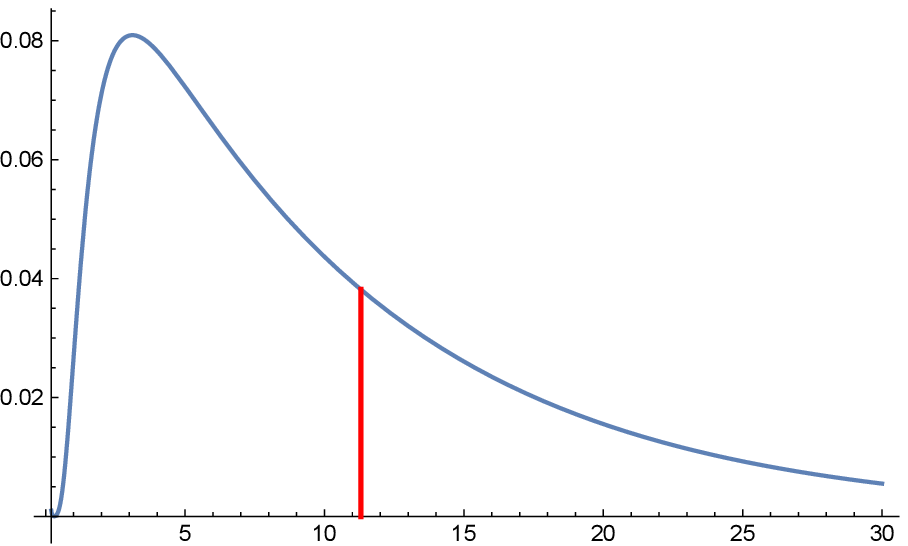}\\
$p_{\tau^*_b}(t;r), 0 < t < 20,$ & $p_{\tau^*_s}(t;r_b), 0 < t < 30$\\ [1em]
\end{tabular}
\begin{tabular}{c}
\includegraphics[width=0.5\textwidth]{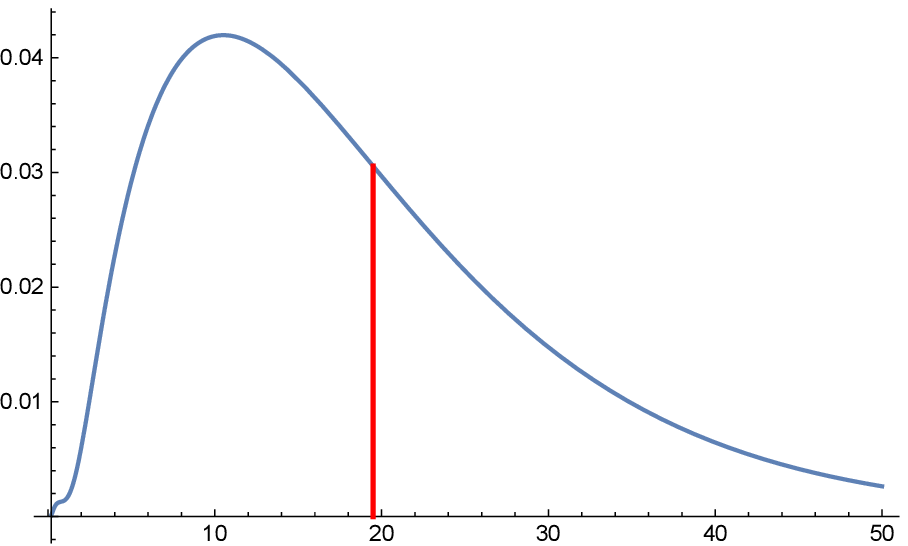}
\\
$p_{\tau^*_b+\tau^*_s}(t;r), 0 < t < 50$
\end{tabular}
\caption{The following plots used $r = 0.08$ and other parameters from \eqref{eq:ex_param}. The selling and buying threshold $r_s$ and $r_b$ are calculated numerically using \eqref{eq:rs} and \eqref{eq:rb} to obtain $r_s \approx 0.026$ and $r_b \approx 0.167$. We plot {the density of the length of time the investor waits before buying} $p_{\tau^*_b}(t;r)$ defined in \eqref{eq:pdf-buy} for $0 < t < 20$ using the first 100 terms of the truncated infinite sum. {The density of the length of time the investor waits before selling} $p_{\tau^*_s}(t;r_b)$ defined in \eqref{eq:pdf-sell} for $0 < t <30$ is also plotted using the first 100 terms of the truncated infinite sum. Lastly, {the density of the total time the investor waits to buy and sell a home} $p_{\tau^*_b+\tau^*_s}(t;r)$ defined in \eqref{eq:pdf-both} for $0 < t < 50$ is plotted using the first 100 indices in the truncated double infinite sum, giving a total of 10000 terms for the approximation. The expectations for each of the random variable calculated in \eqref{eq:all-ave} are shown as a vertical red bar in the respective graphs.}

\label{fig:pdfwait}
\end{figure}
\end{document}